\documentclass[twocolumn,english]{IEEEtran}
\usepackage[T1]{fontenc}
\usepackage{color}
\usepackage{babel}
\usepackage{amsmath}
\usepackage{amssymb}
\usepackage{stackrel}
\usepackage{graphicx}
\usepackage[unicode=true,
 bookmarks=true,bookmarksnumbered=true,bookmarksopen=true,bookmarksopenlevel=1,
 breaklinks=false,pdfborder={0 0 0},pdfborderstyle={},backref=false,colorlinks=false]
 {hyperref}
\hypersetup{pdftitle={Your Title},
 pdfauthor={Your Name},
 pdfpagelayout=OneColumn, pdfnewwindow=true, pdfstartview=XYZ, plainpages=false}
\usepackage{breakurl}

\makeatletter
\ifCLASSOPTIONcompsoc
\usepackage[caption=false,font=normalsize,labelfont=sf,textfont=sf]{subfig}
\else
\usepackage[caption=false,font=footnotesize]{subfig}
\fi
\usepackage{cite}
\usepackage{amsthm}
\newtheorem{theorem}{\textbf{Theorem}}
\newtheorem{corollary}{\textbf{Corollary}}
\newtheorem{lemma}{\textbf{Lemma}}
\newtheorem{proposition}{\textbf{Proposition}}

\usepackage{algorithm}
\usepackage{amsmath}

\@ifundefined{showcaptionsetup}{}{%
 \PassOptionsToPackage{caption=false}{subfig}}
\usepackage{subfig}
\makeatother

\begin{document}

\title{Limitation of SDMA in Ultra-Dense Small Cell Networks}

\author{\IEEEauthorblockN{Junyu Liu, Min Sheng, Jiandong Li}\\
\IEEEauthorblockA{State Key Laboratory of Integrated Service Networks,
Xidian University, Xi'an, Shaanxi, 710071, China\\
Email: junyuliu@xidian.edu.cn, \{msheng, jdli\}@mail.xidian.edu.cn}}
\maketitle
\begin{abstract}
Benefitting from multi-user gain brought by multi-antenna techniques,
space division multiple access (SDMA) is capable of significantly
enhancing spatial throughput (ST) in wireless networks. Nevertheless,
we show in this letter that, even when SDMA is applied, ST would diminish
to be zero in ultra-dense networks (UDN), where small cell base stations
(BSs) are fully densified. More importantly, we compare the performance
of SDMA, single-user beamforming (SU-BF) (one user is served in each
cell) and full SDMA (the number of served users equals the number
of equipped antennas). Surprisingly, it is shown that SU-BF achieves
the highest ST and critical density, beyond which ST starts to degrade,
in UDN. The results in this work could shed light on the fundamental
limitation of SDMA in UDN.
\end{abstract}

\section{Introduction\label{sec:Introduction}}

While network densification is promising to improve network capacity,
its limitation has been manifest in recent research \cite{MUPM_Ref_Original,Ref_Rev_1,Ref_SBPM}.
Depending on practical system settings, it is shown that area spectral
efficiency would decrease as base station (BS) density grows in downlink
cellular networks \cite{Ref_Rev_1}. Worsestill, network capacity
is shown to diminish to zero when single-antenna BSs are over deployed
\cite{Ref_SBPM}. To rejuvenate the potential of network densification,
space division multiple access (SDMA), which is capable of simultaneously
serving multiple users within one cell over identical spectrum resources,
is shown to be of great potential \cite{Ref_MIMO_2,Ref_UDN_MISO}.
Especially, it is shown that single-user beamforming (SU-BF), a special
case of SDMA, could improve network capacity by tens of folds in ultra-dense
networks (UDN), compared to the single-antenna regime \cite{Ref_UDN_MISO}.
However, only single user could be served within one cell using SU-BF.
Therefore, it is imperative to study whether the multi-user gain of
SDMA could be harvested to further improve network capacity in UDN.

In this work, we investigate the performance of SDMA in ultra-dense
downlink small cell networks. As more than one user is served within
one cell using SDMA, spatial resources could be fully exploited. However,
it is unexpectedly shown that network spatial throughput (ST) would
monotonously decrease with the number of served users within one cell
when BS density is sufficiently large, which radically contradicts
with the sparse deployment cases. Moreover, it is proved that applying
SU-BF could as well maximize the critical density, beyond which network
capacity starts to diminish. Therefore, the results reveal the fundamental
limits of SDMA and confirm the superiority of SU-BF in UDN\footnote{Different from \cite{Ref_SBPM}, which investigates the limitation
of network densification, we intend to unveil the limitation of SDMA
in UDN in this work.}.

\section{System Model\label{sec:System-Model}}

We consider a\textcolor{black}{{} downlink small cell network, where
the locations of $N_{\mathrm{a}}$-antenna BSs (with constant transmit
power $P$) and single-antenna users in a two-dimension plane are
modeled as two independent homogeneous Poisson Point Processes (PPPs)
$\Pi_{\mathrm{BS}}=\left\{ \mathrm{BS}_{i}\right\} $ $\left(i\in\mathbb{N}\right)$
of density $\lambda$ and $\Pi_{\mathrm{U}}=\left\{ \mathrm{U}_{j}\right\} $
$\left(j\in\mathbb{N}\right)$ of density $\lambda_{\mathrm{U}}$,
respectively. If $r$ is the 2-dimension distance of cellular link,
the distance between the antennas of BSs and the associated users
is $d=\sqrt{r^{2}+\Delta h^{2}}$, where $\Delta h>0$ denotes the
antenna height difference (AHD) between BSs and users}

\textcolor{black}{We further assume that each user connects to the
BS providing the smallest pathloss. Accordingly, the probability density
function (PDF) of $r$ is given by \cite{book_stochastic_geometry}
\begin{align}
f_{r}\left(x\right)= & 2\pi\lambda xe^{-\pi\lambda x^{2}}.\:\left(x\geq0\right)\label{eq:distance PDF}
\end{align}
Denote $N_{\mathrm{U}}$ $\left(N_{\mathrm{U}}\leq N_{\mathrm{a}}\right)$
as the number of active users within one cell. Meanwhile, a dense
user deployment is considered, i.e., $\lambda\gg\lambda_{\mathrm{U}}$,
such that each full-buffer BS could select at most $N_{\mathrm{a}}$
users to serve using SDMA.}

Channel gain consists of a pathloss component and a small-scale fading
component. In particular, given transmission distance $d$, a generalized
multi-slope pathloss model (MSPM) is adopted to characterize differentiated
signal power attenuation rates within different regions, i.e.,
\begin{equation}
l_{N}\left(\left\{ \alpha_{n}\right\} _{n=0}^{N-1};d\right)=K_{n}d^{-\alpha_{n}},\:R_{n}\leq d<R_{n+1}\label{eq:MUPM}
\end{equation}
where $\alpha_{n}$ denotes the pathloss exponent, $K_{0}=1$, $K_{n}=\prod_{i=1}^{n}R_{i}^{\alpha_{i}-\alpha_{i-1}}$
$\left(n\geq1\right)$, and $0=R_{0}<R_{1}<\cdots<R_{N}=\infty$.
Note that $\alpha_{i}\leq\alpha_{j}$ $\left(i<j\right)$ and $\alpha_{N-1}>2$
for practical concerns. Rayleigh fading is used to model small-scale
fading for mathematical tractability. The rationality of Rayleigh
fading assumption in UDN has been verified via experimental results
in \cite{Ref_Test_BPM}. 

SDMA with zero-forcing precoding is applied, where perfect channel
state information (CSI) from the BSs to the serving downlink users
could be obtained to design the precoders. According to \cite{Ref_MIMO_2,Ref_Gamma_1,Ref_Gamma_2},
under Rayleigh fading channels, the channel power gain caused by small-scaling
fading of both direct link and interfering link follows gamma distribution
when zero-forcing precoding is applied by the multi-antenna technique.
If denoting $H_{\mathrm{U}_{0},\mathrm{BS}_{0}}$ as the channel power
gain caused by small-scale fading from $\mathrm{BS}_{0}$ to $\mathrm{U}_{0}$
(typical pair) and $H_{\mathrm{U}_{0},\mathrm{BS}_{i}}$ $\left(i\neq0\right)$
as the channel power gain caused by small scaling fading from $\mathrm{BS}_{i}$
(interfering BS) to $\mathrm{U}_{0}$, then $H_{\mathrm{U}_{0},\mathrm{BS}_{0}}\sim\Gamma\left(N_{\mathrm{a}}-N_{\mathrm{U}}+1,1\right)$
and $H_{\mathrm{U}_{0},\mathrm{BS}_{i}}\sim\Gamma\left(N_{\mathrm{U}},1\right)$.
It is worth noting that SDMA degenerates into full SDMA when $N_{\mathrm{U}}=N_{\mathrm{a}}$,
while degenerates into SU-BF when $N_{\mathrm{U}}=1$.

We use network ST for performance evaluation of SDMA in UDN. In particular,
ST is defined by 
\begin{equation}
\mathsf{ST}_{N}\left(\lambda\right)=N_{\mathrm{U}}\lambda\mathsf{CP}_{N}\left(\lambda\right)\log_{2}\left(1+\tau\right).\label{eq: define ST}
\end{equation}
In (\ref{eq: define ST}), $\tau$ denotes the decoding threshold
and $\mathsf{CP}_{N}\left(\lambda\right)$ denotes the coverage probability
(CP) of the typical downlink users, which is given by
\begin{equation}
\mathsf{CP}_{N}\left(\lambda\right)=\mathbb{P}\left\{ \mathsf{SIR}_{\mathrm{U}_{0}}>\tau\right\} .\label{eq: define CP}
\end{equation}
In (\ref{eq: define CP}), $\mathsf{SIR}_{\mathrm{U}_{0}}$ denotes
the signal-to-interference ratio (SIR) at $\mathrm{U}_{0}$. Note
that the subscript `$N$' denotes the number of slopes in (\ref{eq:MUPM}).

\textbf{Notation}: Denoting $_{2}F_{1}\left(\cdot,\cdot,\cdot,\cdot\right)$
as Gaussian hypergeometric function, we denote $\omega\left(x,y,z\right)={}_{2}F_{1}\left(x,-\frac{2}{y},1-\frac{2}{y},-z\right)$
in the remaining parts. Besides, $l_{N}\left(\left\{ \alpha_{n}\right\} _{n=0}^{N-1};d\right)$
is replaced by $l_{N}\left(d\right)$ for notation simplicity. 

\section{Performance Evaluation of SDMA in UDN\label{sec:SDMA in UDN}}

In this section, we evaluate the performance of SDMA in terms of network
ST, which depends on the distribution of $\mathsf{SIR}_{\mathrm{U}_{0}}$
according to (\ref{eq: define ST}). In particular, $\mathsf{SIR}_{\mathrm{U}_{0}}$
is given by
\begin{align}
\mathsf{SIR}_{\mathrm{U}_{0}}= & PH_{\mathrm{U}_{0},\mathrm{BS}_{0}}l\left(d_{0}\right)/I_{\mathrm{IC}},\label{eq:SIR expression}
\end{align}
where $I_{\mathrm{IC}}=\underset{\tiny{\mathrm{BS}_{i}\in\tilde{\Pi}_{\mathrm{BS}}}}{\sum}PH_{\mathrm{U}_{0},\mathrm{BS}_{i}}l\left(d_{i}\right)$
is the inter-cell interference, $d_{i}$ denotes the distance between
the antennas of $\mathrm{BS}_{i}$ and $\mathrm{U}_{0}$, and $H_{\mathrm{U}_{0},\mathrm{BS}_{i}}$
denotes the power gain caused by small-scale-fading from $\mathrm{BS}_{i}$
to $\mathrm{U}_{0}$. Following $\mathsf{SIR}_{\mathrm{U}_{0}}$ in
(\ref{eq:SIR expression}), we present the results on CP and ST in
the following corollary.

\begin{corollary}[CP and ST in SDMA System]

When SDMA is applied under MSPM, ST in downlink small cell networks
is given by $\mathsf{ST}_{N}\left(\lambda\right)=N_{\mathrm{U}}\lambda\mathsf{CP}_{N}\left(\lambda\right)\log_{2}\left(1+\tau\right)$,
where
\begin{align}
\mathsf{CP}_{N}\left(\lambda\right)= & \mathbb{E}\left[\stackrel[k=0]{N_{\mathrm{a}}-N_{\mathrm{U}}}{\sum}\frac{\left(-s\right)^{k}}{k!}\frac{\mathrm{d}^{k}}{\mathrm{d}s^{k}}\mathcal{L}_{I_{\mathrm{IC}}}\left(s\right)\left|_{s=\frac{\tau}{Pl_{N}\left(d_{0}\right)}}\right.\right].\label{eq:CP SDMA}
\end{align}
In (\ref{eq:CP SDMA}), $\mathcal{L}_{I_{\mathrm{IC}}}\left(s\right)$
is the Laplace Transform of $I_{\mathrm{IC}}$ evaluated at $s=\frac{\tau}{Pl_{N}\left(d_{0}\right)}$,
which is given by
\begin{align}
\mathcal{L}_{I_{\mathrm{IC}}}\left(s\right)= & e^{-2\pi\lambda\int_{d_{0}}^{\infty}x\left(1-\left(1+sPl_{N}\left(x\right)\right)^{-1}\right)\mathrm{d}x}.\label{eq:Laplace of Interference}
\end{align}

\label{corollary: CP and ST SDMA}

\end{corollary}

\begin{proof}Following \cite[Corollary 2]{Ref_UDN_MISO}, the results
can be obtained with easy manipulation. The detail is omitted due
to space limitation.\end{proof}

Corollary \ref{corollary: CP and ST SDMA} provides a numerical approach
to evaluate the performance of SDMA in UDN. As well, the impact of
key parameters of SDMA, such as the number of antennas and served
users, on system performance could be captured. In particular, we
plot ST as a function of BS density in Fig. \ref{fig:ST SSPM and DSPM}
under single-slope pathloss model (SSPM), i.e., $N=1$ in (\ref{eq:MUPM}),
and dual-slope pathloss model (DSPM), i.e., $N=2$ in (\ref{eq:MUPM}).
It can be seen from Fig. \ref{fig:ST - impact of Na} that network
ST could be significantly improved by SDMA, compared to the single-antenna
regime. Especially, the maximal ST could be improved by 21 folds under
DSPM when 16 antennas are equipped by each BS and 2 users are served
within one cell. Meanwhile, we evaluate the impact of AHD on the network
ST in Fig. \ref{fig:ST - impact of AHD}. It is observed that network
ST would be significantly over-estimated without considering the AHD,
i.e., ST even linearly increases with BS density. As the $\Delta h=0$m
assumption is no longer impractical when the distance from transmitters
and the intended receivers is small, this confirms the importance
of modeling AHD when evaluating the performance of UDN.

Taking $\Delta h>0$m into account, however, it is pessimistic to
observe that ST would asymptotically approach zero under dense deployment
even when SDMA is applied. Therefore, we analytically study ST scaling
law next.

\begin{figure}[t]
\begin{centering}
\subfloat[\label{fig:ST - impact of Na}Impact of $N_{\mathrm{a}}$.]{\begin{centering}
\includegraphics[width=3in]{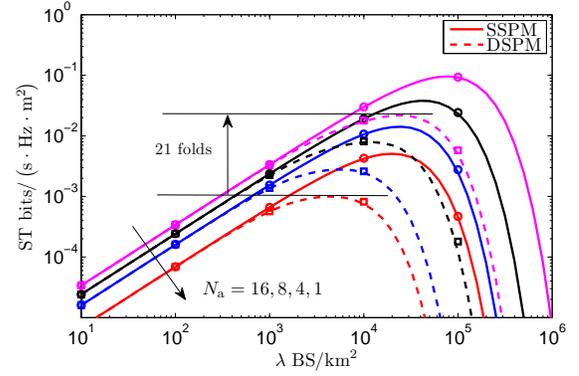}
\par\end{centering}
}
\par\end{centering}
\begin{centering}
\subfloat[\label{fig:ST - impact of AHD}Impact of AHD $\Delta h$.]{\begin{centering}
\includegraphics[width=3in]{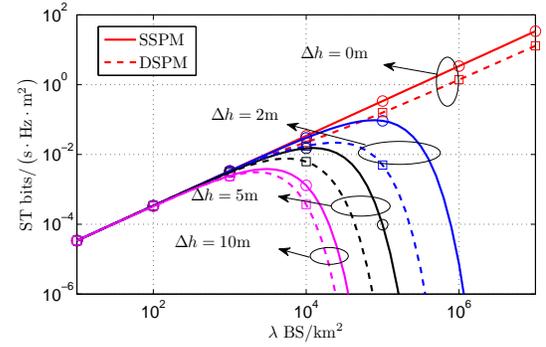}
\par\end{centering}
}
\par\end{centering}
\caption{\textcolor{blue}{\label{fig:ST SSPM and DSPM}}ST scaling laws. Set
$N_{\mathrm{U}}=1$ when $N_{\mathrm{a}}=1$, while set $N_{\mathrm{U}}=2$
when $N_{\mathrm{a}}>1$. For system settings, set $P=23$dBm and
$\tau=10$dB. Set $\alpha_{0}=4$ for SSPM, and $\alpha_{0}=2.5$,
$\alpha_{1}=4$ and $R_{1}=10$m for DSPM. In (a), set $\Delta h=2$m.
In (b), set $N_{\mathrm{a}}$=16. Lines and markers denote numerical
and simulation results, respectively.}
\end{figure}

\begin{corollary}[CP and ST Scaling Laws in SDMA System]

In SDMA system, CP and ST scale with BS density $\lambda$ as $\mathsf{CP}_{N}\left(\lambda\right)\sim e^{-\kappa\lambda}$
and $\mathsf{ST}_{N}\left(\lambda\right)\sim\lambda e^{-\kappa\lambda}$
($\kappa$ is a constant), respectively.

\label{corollary: CP and ST scaling law SDMA}

\end{corollary}

\begin{proof}Given the number of antennas, the experience of single
user would be degraded if more users within one cell share the antenna
resources. Therefore, we have $\mathsf{CP}_{N}^{\mathrm{F}}\left(\lambda\right)\leq\mathsf{CP}_{N}\left(\lambda\right)\leq\mathsf{CP}_{N}^{\mathrm{S}}\left(\lambda\right)$,
where $\mathsf{CP}_{N}^{\mathrm{F}}\left(\lambda\right)$ is obtained
by setting $N_{\mathrm{U}}=N_{\mathrm{a}}$ (full SDMA) and $\mathsf{CP}_{N}^{\mathrm{S}}\left(\lambda\right)$
is obtained by setting $N_{\mathrm{U}}=1$ (SU-BF). Hence, $\mathsf{CP}^{\mathrm{F}}\left(\lambda\right)$\begin{small}
\begin{align}
\overset{\left(\mathrm{a}\right)}{=} & \mathbb{E}_{r_{0}}\left[\exp\left(-2\pi\lambda\int_{d_{0}}^{\infty}x\left(1-\left(1+sPl_{N}\left(x\right)\right)^{-N_{\mathrm{a}}}\right)\mathrm{d}x\right)\right]\nonumber \\
> & \mathbb{E}_{r_{0}>R_{N-1}}\left[\exp\left(-2\pi\lambda\int_{d_{0}}^{\infty}x\left(1-\left(1+sPl_{N}\left(x\right)\right)^{-N_{\mathrm{a}}}\right)\mathrm{d}x\right)\right]\nonumber \\
\overset{\left(\mathrm{b}\right)}{=} & \mathbb{E}_{r_{0}>R_{N-1}}\left[\exp\left(-\pi\lambda d_{0}^{2}\left(\omega\left(N_{\mathrm{a}},\alpha_{N-1},\tau\right)-1\right)\right)\right]\nonumber \\
\overset{\left(\mathrm{c}\right)}{=} & \frac{\exp\left[-\pi\lambda\left(\omega\left(N_{\mathrm{a}},\alpha_{N-1},\tau\right)\left(R_{N-1}^{2}+\Delta h^{2}\right)-\Delta h^{2}\right)\right]}{\omega\left(N_{\mathrm{a}},\alpha_{N-1},\tau\right)}\nonumber \\
= & \mathsf{CP}_{N}^{\mathrm{L}}\left(\lambda\right),\label{eq:scaling law 1}
\end{align}
\end{small}where $s=\frac{\tau}{Pl_{N}\left(d_{0}\right)}$. In (\ref{eq:scaling law 1}),
(a) follows due to $H_{\mathrm{U}_{0},\mathrm{BS}_{0}}\sim\Gamma\left(1,1\right)$
since $N_{\mathrm{a}}=N_{\mathrm{U}}$, (b) follows due to $s=\frac{\tau}{PK_{N-1}d_{0}^{-\alpha_{N-1}}}$
and $l_{N}\left(x\right)=K_{N-1}x^{-\alpha_{N-1}}$ when $r_{0}>R_{N-1}$,
and (c) follows due to the PDF of $r_{0}$ given in (\ref{eq:distance PDF}).
In consequence, $\mathsf{CP}_{N}^{\mathrm{L}}\left(\lambda\right)\sim e^{-\kappa\lambda}$
holds. Besides, it follows from \cite[Theorem 2]{Ref_UDN_MISO} that
$\mathsf{CP}^{\mathrm{S}}\left(\lambda\right)\sim e^{-\kappa\lambda}$.
Therefore, $\mathsf{CP}_{N}\left(\lambda\right)\sim e^{-\kappa\lambda}$.

Following the definition of ST, the proof is complete.\end{proof}

Corollary \ref{corollary: CP and ST scaling law SDMA} reveals the
fundamental limitation of SDMA in UDN. Meanwhile, it is intuitive
to obtain that SU-BF outperforms SDMA and full SDMA in terms of CP.
Therefore, it is interesting to analytically compare the performance
of these schemes in terms of ST as well, the detail of which is presented
in the following part.

\section{SDMA, Full SDMA and SU-BF: A Special Case Study\label{sec:SDMA comparison}}

In this section, analytical comparison of SDMA, full SDMA and SU-BF
is made from the ST perspective. Especially, a special case study
is performed under SSPM given by
\begin{align}
l_{1}\left(d\right) & =d^{-\alpha_{0}}.\:d\geq0\label{eq:SSPM}
\end{align}
Nevertheless, the results on ST in Corollary \ref{corollary: CP and ST SDMA}
are much too complicated even when SSPM is applied. To facilitate
the comparison, we first present a simple but effective approximation
on ST in the Proposition \ref{proposition: CP and ST SDMA APP}.

\begin{proposition}

When SDMA is applied under SSPM, ST in ultra-dense downlink small
cell networks can be approximated as $\mathsf{ST}_{1}^{\dagger}\left(\lambda\right)=\lambda\mathsf{CP}_{1}^{\dagger}\left(\lambda\right)\log_{2}\left(1+\tau\right)$,
where
\begin{align}
\mathsf{CP}_{1}^{\dagger}\left(\lambda\right)= & \frac{\exp\left(-\pi\lambda\delta\left(\alpha_{0}\right)\right)}{\omega\left(N_{\mathrm{U}},\alpha_{0},\tau^{\mathrm{S}}\right)},\label{eq:CP SDMA APP}
\end{align}
where $\delta\left(x\right)=\Delta h^{2}\left(\omega\left(N_{\mathrm{U}},x,\tau^{\mathrm{S}}\right)-1\right)$
and $\tau^{\mathrm{S}}=\frac{\tau}{N_{\mathrm{a}}-N_{\mathrm{U}}+1}$.

\label{proposition: CP and ST SDMA APP}

\end{proposition}

\begin{proof}The key to the approximation is to use $G_{\mathrm{U}_{0},\mathrm{BS}_{0}}\sim\mathrm{Exp}\left(\frac{1}{N_{\mathrm{a}}-N_{\mathrm{U}}+1}\right)$
to replace $H_{\mathrm{U}_{0},\mathrm{BS}_{0}}\sim\varGamma\left(N_{\mathrm{a}}-N_{\mathrm{U}}+1,1\right)$.
Accordingly, we have
\begin{align}
\mathsf{CP}_{1}^{\dagger}\left(\lambda\right)= & \mathbb{E}_{r_{0}}\left[e^{-2\pi\lambda\int_{d_{0}}^{\infty}x\left(1-\left(1+s^{\dagger}Pl_{1}\left(x\right)\right)^{-N_{\mathrm{U}}}\right)\mathrm{d}x}\right]\nonumber \\
= & \mathbb{E}_{r_{0}}\left[e^{-\pi\lambda d_{0}^{2}\left[\omega\left(N_{\mathrm{U}},\alpha_{0},\tau^{\mathrm{S}}\right)-1\right]}\right],\label{eq:APP proof 1}
\end{align}
where $s^{\dagger}=\frac{\tau^{\mathrm{S}}}{Pl_{1}\left(d_{0}\right)}$
and $d_{0}=\sqrt{r_{0}^{2}+\Delta h^{2}}$. Aided by (\ref{eq:APP proof 1})
and the PDF of $r_{0}$ given in (\ref{eq:distance PDF}), $\mathsf{CP}_{1}\left(\lambda\right)$
can be obtained.\end{proof}

\begin{figure}[t]
\begin{centering}
\subfloat[\label{fig:test accuracy}Examine the approximation accuracy.]{\begin{centering}
\includegraphics[width=3in]{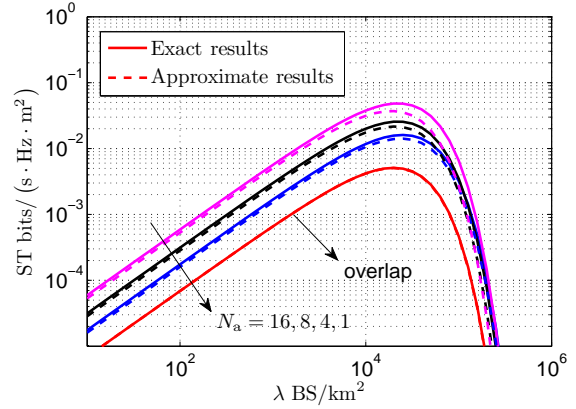}
\par\end{centering}
}
\par\end{centering}
\begin{centering}
\subfloat[\label{fig:comparison}Comparison of SDMA, full SDMA and SU-BF.]{\begin{centering}
\includegraphics[width=3in]{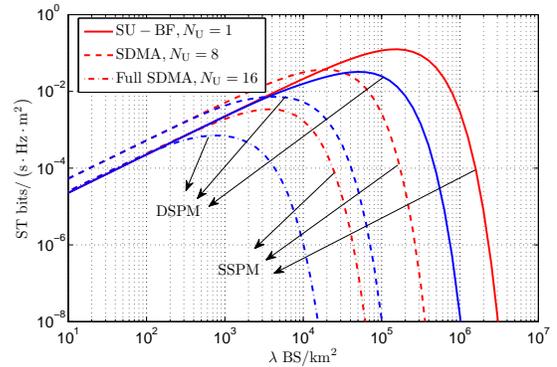}
\par\end{centering}
}
\par\end{centering}
\caption{\label{fig:SDMA VS MISO}ST scaling laws. In (a), set $N_{\mathrm{U}}=1$
when $N_{\mathrm{a}}=1$ and set $N_{\mathrm{U}}=N_{\mathrm{a}}/2$
when $N_{\mathrm{a}}>1$. In (b), set $N_{\mathrm{a}}$=16 and other
system parameters are identical to those in Fig. \ref{fig:ST SSPM and DSPM}.}
\end{figure}

We next verify the accuracy of the approximation in Proposition \ref{proposition: CP and ST SDMA APP}
using Fig. \ref{fig:test accuracy}. It can be seen that the gaps
between the exact and approximate results are small. They even overlap
when $N_{\mathrm{a}}=1$. Moreover, it is observed that the critical
densities obtained via the approximate and exact results are almost
identical. Therefore, it is valid to use the approximate results to
derive critical density as follows.

\begin{corollary}

The critical density in SDMA downlink small cell networks is given
by 
\begin{align}
\lambda^{*}= & \frac{1}{\pi\Delta h^{2}\left(\omega\left(N_{\mathrm{U}},\alpha_{0},\tau^{\mathrm{S}}\right)-1\right)},\label{eq:CP SDMA APP-1}
\end{align}
where $\tau^{\mathrm{S}}=\frac{\tau}{N_{\mathrm{a}}-N_{\mathrm{U}}+1}$.

\label{corollary: critical density}

\end{corollary}

\begin{proof}The proof can be feasibly completed by solving $\frac{\partial\mathsf{ST}_{1}^{\dagger}\left(\lambda\right)}{\partial\lambda}=0$.\end{proof}

Corollary \ref{corollary: critical density} identifies the relationship
of key system parameters and critical density. For instance, it is
apparent that $\lambda^{*}$ inversely increases with $\Delta h^{2}$.
For this reason, it is suggested to lower the transmitter and receiver
antenna heights and reduce the AHD such that network densification
could be still effective in enhancing spectrum efficiency even when
BS density is large.

Next, we use the results in Corollary \ref{corollary: critical density}
to further evaluate the performance of SDMA, full SDMA and SU-BF in
UDN. Before that, the results on $\omega\left(x,y,z\right)$ are given
in Lemma \ref{lemma: hypergeometric function}.

\begin{lemma}

Given $N>1$ (an integer), $-1<b<0$, $c>0$ and $z<0$, $_{2}F_{1}\left(N,b,c,z\right)$
is an increasing function of $N$.

\label{lemma: hypergeometric function}

\end{lemma}

\begin{proof}The proof can be completed using mathematical induction.
Assuming $\varDelta_{1}={}_{2}F_{1}\left(N+1,b,c,z\right)-{}_{2}F_{1}\left(N,b,c,z\right)=\frac{bz}{c}{}_{2}F_{1}\left(N+1,b+1,c+1,z\right)>0$,
if the inequality $\varDelta_{2}={}_{2}F_{1}\left(N+2,b,c,z\right)-{}_{2}F_{1}\left(N+1,b,c,z\right)>0$
holds, the proof is complete. By definition of hypergemetric function,
we have\begin{small}
\begin{align}
\varDelta_{2} & =\frac{bz}{c}{}_{2}F_{1}\left(N+2,b+1,c+1,z\right)\nonumber \\
 & =\frac{b}{N+1}\left[_{2}F_{1}\left(N+1,b+1,c,z\right)-{}_{2}F_{1}\left(N+1,b,c,z\right)\right].\label{eq:hyper proof 1}
\end{align}
\end{small}As $N>0$ and $b<0$, to prove $\varDelta_{2}>0$, it
is sufficient to show $_{2}F_{1}\left(N,b+1,c,z\right)-{}_{2}F_{1}\left(N,b,c,z\right)<0$.
Hence, we have
\begin{align}
\varDelta_{3}= & _{2}F_{1}\left(b+1,N,c,z\right)-{}_{2}F_{1}\left(b,N,c,z\right)\nonumber \\
= & \frac{Nz}{c}{}_{2}F_{1}\left(b+1,N+1,c+1,z\right)\nonumber \\
= & \frac{Nz}{c}{}_{2}F_{1}\left(N+1,b+1,c+1,z\right).\label{eq:hyper proof 2}
\end{align}
According to the assumption that $\varDelta_{1}>0$, we have $_{2}F_{1}\left(N+1,b+1,c+1,z\right)>0$.
Therefore, $\varDelta_{3}<0$ holds.\end{proof}

Aided by Lemma \ref{lemma: hypergeometric function}, we provide the
results on full SDMA in the following theorem.

\begin{theorem}

When full SDMA is applied, the critical density of downlink small
cell networks is a decreasing function of the number of antennas.

\label{theorem: critical density full SDMA}

\end{theorem}

\begin{proof}Applying full SDMA, $N_{\mathrm{U}}=N_{\mathrm{a}}$
and the critical density in Corollary \ref{corollary: critical density}
degenerates into
\begin{align}
\lambda^{*}= & \frac{1}{\pi\Delta h^{2}\left(\omega\left(N_{\mathrm{a}},\alpha_{0},\tau\right)-1\right)}.\label{eq:critical density proof 1}
\end{align}
Since $\omega\left(N_{\mathrm{a}},\alpha_{0},\tau\right)$ is an increasing
function of $N_{\mathrm{a}}$ (Lemma \ref{lemma: hypergeometric function}),
the proof is complete.\end{proof}

Theorem \ref{theorem: critical density full SDMA} evaluates the performance
of full SDMA in terms of critical density. More specifically, it is
observed from Fig. \ref{fig:comparison} that ST would begin to decrease
at a smaller critical density when more antennas are equipped. Following
Theorem \ref{theorem: critical density full SDMA}, we further compare
SDMA, full SDMA and SU-BF in Theorem \ref{theorem: critical density SDMA}.

\begin{theorem}

Given the number of antennas equipped on each BS, the critical density
achieved by SU-BF is greater than those achieved by SDMA and full
SDMA.

\label{theorem: critical density SDMA}

\end{theorem}

\begin{proof}The results can be proved by showing $\lambda^{*}$
in Corollary \ref{corollary: critical density} is a decreasing function
of $N_{\mathrm{U}}$ or equivalently $\omega\left(N_{\mathrm{U}},\alpha_{0},\tau^{\mathrm{S}}\right)$
is an increasing function of $N_{\mathrm{U}}$. By making an extension
of \cite[Lemma 1]{Ref_SBPM}, it can be shown that $\omega\left(N_{\mathrm{U}},\alpha_{0},\tau^{\mathrm{S}}\right)$
increases with $\tau^{\mathrm{S}}=\frac{1}{N_{\mathrm{a}}-N_{\mathrm{U}}+1}$,
which monotonously increases with $N_{\mathrm{U}}$. Therefore, aided
by Lemma \ref{lemma: hypergeometric function}, it can be proved that
$\omega\left(N_{\mathrm{U}},\alpha_{0},\tau^{\mathrm{S}}\right)$
increases with $N_{\mathrm{U}}$.\end{proof}

Theorem \ref{theorem: critical density SDMA} indicates that simultaneously
serving multiple users would enable ST to decrease at a smaller BS
density. In other words, compared to SU-BF, SDMA potentially degrades
spectrum efficiency in UDN. This is inconsistent with traditional
understanding of SDMA in sparse deployment \cite{Ref_MIMO_2}, which
shows higher spectrum efficiency could be achieved by SDMA via serving
more users. This can be verified via the results in Fig. \ref{fig:comparison}.
Specifically, although a higher ST could be obtained by SDMA when
$N_{\mathrm{U}}$ is large in sparse scenario, it would experience
a notable decrease at a greater BS density, compared to SU-BF. The
intuition behind this is that the multi-user gain harvested by SDMA
is ruined by the overwhelming interference in UDN. 

Besides, it is numerically shown in Fig. \ref{fig:comparison} that
SU-BF outperforms SDMA and full SDMA in terms of ST in UDN as well.
From this perspective, SU-BF is more favorable to enhance network
capacity in UDN. While the analytical comparison is made based on
SSPM, it could be easily tested that the comparison results are still
valid under MSPM. For instance, the results under DSPM, i.e., $N=2$
in (\ref{eq:MUPM}), are verified in Fig. \ref{fig:comparison}.

\section{Conclusion\label{sec:Conclusion}}

In this letter, the performance of SDMA has been explored in ultra-dense
downlink small cell networks. Specifically, SDMA, albeit incapable
of improving network capacity scaling law, is proved to significantly
enhancing spatial throughput, compared to single-antenna regime. Moreover,
it is interestingly shown that network capacity and critical density
could be maximized when single user is served in each small cell,
which contradicts with the intuition. The primary reason is that demerits
of overwhelming inter-cell interference dominate the benefits of multi-user
gain of SDMA in UDN. Therefore, the outcomes of this work could provide
guideline towards the application and optimization of SDMA in UDN
via reasonably selecting the number of users to serve. 

\bibliographystyle{IEEEtran}
\bibliography{ref_BPM}

\end{document}